\documentclass[10pt]{llncs}
\usepackage{amssymb,amsmath}
\usepackage{epsfig,graphics}
\usepackage{llncsdoc}
\usepackage{enumerate}

\newcommand{\abs}[1]{\lvert #1 \rvert}
\newcommand{\ldotdot}{ \, . \, . \, }

\newtheorem{observation}[theorem]{Observation}
\begin{document}

\title{On Hardness of Jumbled Indexing}
\author{Amihood Amir\inst{1}\thanks{Supported by NSF grant CCR-09-04581, ISF grant 347/09, and BSF grant 2008217.} \and Timothy M. Chan\inst{2}\thanks{Part of the work of this author was done while visiting the Department of Computer Science and Engineering,
Hong Kong University of Science and Technology.}
 \and Moshe Lewenstein\thanks{This research was done in part while the author was on sabbatical in the U. of Waterloo. The research is supported by BSF grant 2010437 and GIF grant 1147/2011.}\inst{3} \and Noa Lewenstein\inst{4}}
\institute{Bar-Ilan University and Johns Hopkins University \and University of Waterloo \and Bar-Ilan University \and Netanya College}

\maketitle

\sloppy

\begin{abstract}
Jumbled indexing is the problem of indexing a text $T$ for queries that ask whether there is a substring of $T$ matching a pattern represented as a Parikh vector, i.e., the vector of frequency counts for each character. Jumbled indexing has garnered a lot of interest in the last four years; for a partial list see~\cite{ABP14,BFKL13,BCFL10,CFL09,CLWY12,DMMT14,GHLW13,GG13,HLRW14,KRR13,MR10,MR12}. There is a naive algorithm that preprocesses all answers in $O(n^2|\Sigma|)$ time  allowing quick queries afterwards, and there is another naive algorithm that requires no preprocessing but has $O(n\log|\Sigma|)$ query time. Despite a tremendous amount of effort there has been little improvement over these running times.

In this paper we provide good reason for this. We show that, under a 3SUM-hardness assumption, jumbled indexing for alphabets of size $\omega(1)$
requires $\Omega(n^{2-\epsilon})$ preprocessing time or $\Omega(n^{1-\delta})$ query time for any $\epsilon,\delta>0$. In fact, under a stronger 3SUM-hardness assumption, for any constant alphabet size $r\ge 3$ there exist describable fixed constant $\epsilon_r$ and $\delta_r$ such that jumbled indexing requires $\Omega(n^{2-\epsilon_r})$ preprocessing time or $\Omega(n^{1-\delta_r})$ query time.

\end{abstract}


\section{Introduction}

Equal length strings are said to {\em jumble-match} if they are commutatively equivalent (sometimes called Abelian equivalent), i.e., if one string can be obtained from the other by permuting its characters. A jumble match can be described using {\em Parikh vectors} which are vectors maintaining the frequency count of each alphabet character. Two strings jumble-match if they have the same Parikh vector. We also say that a string jumble-matches a Parikh vector $\psi$ if the string's Parikh vector is the same as $\psi$.

Parikh vectors were introduced in~\cite{Parikh66} and have been used to analyze grammars~\cite{KT10} and characterize commutative languages~\cite{Holub08}. Furthermore, jumbled pattern matching appears in various applications of computational biology,
such as SNP discovery~\cite{Bocker07}, analysis of similarities among different protein sequences~\cite{HASS04}, and automatic pattern discovery in biosequencing applications~\cite{ELP04}. It has also been examined in the streaming model~\cite{LLZ12}.

Jumbled pattern matching on its own can easily be solved by using a sliding window in linear time for the alphabet $\{1,\ldots,O(n)\}$, or $O(n \log |\Sigma|)$ time for a general alphabet $\Sigma$. In contrast, the exact pattern matching problem can only be solved in linear time via more complex techniques (e.g., see~\cite{KMP77,BM77}).
Jumbled pattern matching has also been studied along with other metrics (e.g., see \cite{BEL04,BLM14,AALS03}).

\subsection{Jumbled Indexing}

{\em Jumbled indexing (JI)}, currently under very active research, asks whether one can ``index'' jumbled matching. The goal is to preprocess a given text $S$ efficiently so that when given a Parikh vector $\psi$ one can quickly check whether there exists a substring of $S$ that jumble-matches $\psi$.

For classical exact matching, text indexing paradigms of linear size and with near linear query time (in the query size) exist since the introduction of suffix trees~\cite{Weiner73}. Many other efficient text indexing structures have been studied since then, such as suffix array~\cite{MM93}. Other matching problems have also been successfully transformed into efficient indexing paradigms. For example, {\em parameterized matching}  allows parametric symbols that are required to map to characters in a consistent manner. Parameterized matching was introduced by Baker~\cite{Baker96,Baker97} for detection of repetitive similar modules in software and has applications for color images~\cite{ACD02,BMK95,SB91} and approximate image search~\cite{HLS07}. Parameterized matching can be solved in linear time~\cite{AFM94}. In~\cite{Baker96} a parameterized suffix tree was introduced. Both the preprocessing and the query times are near-linear (where the latter is linear in the query size). Another example is order-preserving matchings, where two numerical strings match if their order is preserved. Efficient order-preserving matchings were presented in~\cite{KKRRW13} and recently an order-preserving index was introduced~\cite{CIKKLPRRW13} that can also be preprocessed in linear time with linear time queries (in the query size). Indexing with errors~\cite{CGL04} has proven to be somewhat harder.

Given that jumbled matching can be trivially solved in linear time, for the above-mentioned alphabets, one would expect that jumbled indexing would be a relatively easy problem. However, jumbled indexing is surprisingly difficult.

There are two naive methods to solve jumbled indexing. One is to use the sliding window technique mentioned above for every query that arrives. This can be done in $O(n)$ time if the alphabet is a subset of $[n]$, where $n$ is the text size. Another method is to preprocess all possible answers in advance by computing the Parikh vectors of every substring in $O(n^2|\Sigma|)$ time. Improving upon this has proven to be challenging even for constant-sized alphabets.

In an effort to make progress on JI the simplest version of the problem was considered, that of a binary alphabet. A neat property of a binary alphabet is that a Parikh vector $(i,j)$ appears in text $T$ iff $i$ is between the minimum and maximum number of $1$s over all substrings of length $i+j$. This was used in~\cite{CFL09} to obtain efficient query time by storing the minimum and maximum values of all possible lengths, yielding an index of $O(n)$ space and $O(1)$ query time. However, the preprocessing still took $O(n^2)$ time.

Burcsi et al.~\cite{BCFL10} and, independently, Moosa and Rahman~\cite{MR10} succeeded in improving the preprocessing time by a $\log$ factor to $O({n^2 \over \log n})$. They achieved this by reducing binary JI to (min,+)-convolution which can be solved in $O({n^2 \over \log n})$ time~\cite{BCDEHILT06}. Later, Moosa and Rahman~\cite{MR12} improved this to $O({n^2 \over \log^2 n})$ by using the four-Russians trick. Recently, Hermelin et al.~\cite{HLRW14} reduced the problem to (min,+)-matrix multiplication or all-pairs shortest paths,
but a similar reduction has already appeared in an earlier paper by Bremner et al.~\cite{BCDEHILT06};
with the latest breakthrough by Williams~\cite{Ryan13} on all-pairs shortest paths,
the preprocessing time for binary JI becomes  $O({n^2 \over 2^{\Omega((\log n / \log\log n)^{0.5})}})$.
 For the binary case there are also algorithms for run-length encoded strings~\cite{BFKL13,GG13} and for an approximate version of the problem~\cite{CLWY12}. The binary case was also extended to trees~\cite{GHLW13}.

Lately, there has been some progress also for non-binary alphabets. Kociumaka et al.~\cite{KRR13}
presented a solution  for JI for any constant-sized alphabet $\Sigma$ that uses $O({n^2\log^2\log n \over \log n})$ preprocessing time and space and answers queries in $O(({\log n \over \log\log n})^{2|\Sigma|-1})$ time. Amir et al.~\cite{ABP14} proposed a solution for constant-sized alphabets that preprocesses in $O(n^{1+\epsilon})$ time and answers queries in $\tilde{O}(m^{1 \over \epsilon})$ time, where $m$ is the sum of the Parikh vector elements. In an even newer paper, Durocher et al.~\cite{DMMT14} considered alphabet size $|\Sigma| = o(({\log n \over \log\log n})^2)$ and showed how to construct an index in $O(|\Sigma|({n\over \log_{|\Sigma|}n})^2)$ time and answer queries in $O(n^{\epsilon} + |\Sigma|)$ time, where $\epsilon>0$ is an arbitrary small constant.
This still leaves us in a sad state of affairs. In all the (exact) solutions mentioned
for $|\Sigma|\ge 3$ the time complexity of preprocessing or the time complexity of querying is always within polylogarithmic factors of one of the above two naive algorithms.
The question that has troubled the community in these last few years is whether jumbled indexing could be solved with $O(n^{2-\epsilon})$ preprocessing time and $O(n^{1-\delta})$ for some constants $\epsilon,\delta>0$.

In this paper we show that for alphabets of $\omega(1)$ size
this is impossible under a 3SUM-hardness assumption. We further show that for any constant alphabet size $r\ge 3$ there exist describable fixed constants $\epsilon_r$ and $\delta_r$ such that jumbled indexing requires $\Omega(n^{2-\epsilon_r})$ preprocessing time or $\Omega(n^{1-\delta_r})$ query time under a stronger 3SUM-hardness assumption.

\subsection{3SUM}

Numerous algorithmic problems have polynomial time upper bounds that we suspect are the best obtainable but proving matching lower bounds is difficult in classical computational models. Recently, a different approach for showing hardness (e.g.~\cite{WW10}) has been to choose an algorithmic problem that seem harder than others, e.g. maximum flow, APSP, edit distance, or 3SUM, and to use them as a hard primitive and reduce them to other problems that we would like to show are hard.

\medskip
In this paper we use the {\em 3SUM problem} defined as follows.

\label{sec:def}

\begin{itemize}
\item{\bf Input:} $x_1,x_2,\hdots, x_n$.

\item{\bf Output:} yes, if distinct $i, j$ and $k$ exist such that $x_i + x_j = x_k$. No, otherwise.
\end{itemize}

As far back as the mid 90's there were reductions from 3SUM, especially within the computational geometry community. Gajentaan and Overmars~\cite{GO95} were the first to reduce from 3SUM in order to provide evidence for near-quadratic
complexity for computational geometry problems such as minimum-area triangle, finding 3 collinear points, and determining whether $n$ axis-aligned rectangles cover a given rectangle. Others followed and quite a few problems are now known to be 3SUM-hard.

P\u{a}tra\c{s}cu~\cite{Patrascu10} pointed out that most of the reductions transform the condition $x_i + x_j = x_k$ into some geometric or algebraic condition by common arithmetic,
but it is difficult to use 3SUM for reductions to purely combinatorial problems, such as those on graphs or strings. To overcome this he defined {\em Convolution-3SUM}, a more restricted 3SUM version, which is just as hard as 3SUM in the sense that an
$O(n^{2-\epsilon})$-time
solution for Convolution-3SUM for some $\epsilon>0$
would imply an $O(n^{2-\epsilon'})$-time solution for 3SUM
for some $\epsilon'>0$~\cite{Patrascu10}.

\medskip

The {\em Convolution-3SUM  problem} is defined as follows.
\begin{itemize}
\item {\bf Input:} $x_1,\hdots, x_n$.

\item{\bf Output:} Yes, if there are distinct $i$ and $j$ such that $x_i + x_j = x_{i+j}$. No, otherwise.
\end{itemize}
By shuffling and changing indices an alternative equivalent output is:

\begin{itemize}
\item{\bf Output:} Yes, if there are distinct $i$ and $j$ such that $x_i - x_j = x_{i-j}$. No, otherwise.
\end{itemize}


We consider these problems in the RAM model with the elements belonging to an integer set $\{-u,\hdots, u\}$ as was assumed by others, e.g.~\cite{BDP05,Patrascu10}.
It is possible to achieve an algorithm of $O(u\log u)$ time for the 3SUM problem~\cite{BDP05} by
Fast Fourier transform. This can easily be transformed into an $O(nu\log(nu))$ time algorithm for Convolution-3SUM\@. P\u{a}tra\c{s}cu~\cite{Patrascu10} pointed out that the techniques of Baran et al.~\cite{BDP05} yield a (randomized) reduction, for the 3SUM problem, from a large domain $\{-u,\hdots, u\}$ to the domain of $\{-n^3,\hdots, n^3\}$. This reduction can be adapted for Convolution-3SUM from $\{-u,\hdots, u\}$ to  $\{-n^2,\hdots, n^2\}$. The reason is that in 3SUM there are $n^3$ triples to consider when bounding the number of false positives, but in Convolution-3SUM there are only $n^2$ triples ($x_i + x_j = x_{i+j}$) to consider.

Hence, Convolution-3SUM for input $\subset \{-u,\hdots, u\}$ is hard if $u \geq n^2$ and is easier than $O(n^2)$ for $u \ll n$. Convolution-3SUM for inputs $\subset \{-n,\hdots, n\}$ seems (though has not proven) to be as hard as the general case. This leads us to state two hardness assumptions. For both we assume, as in~\cite{BDP05,Patrascu10}, the Word RAM model with words of $O(\log n)$ bits.

\begin{itemize}
\item
{\bf 3SUM-hardness assumption:} Any algorithm for Convolution-3SUM requires $n^{2-o(1)}$ time in expectation to determine whether a set $\{x_1,\hdots, x_n\} \subset \{-n^2, \hdots, n^2\}$ contains a pair $x_i, x_j$ such that $x_i - x_j = x_{i-j}$.

\item
{\bf Strong 3SUM-hardness assumption:} Any algorithm for Convolution-3SUM requires $n^{2-o(1)}$ time in expectation to determine whether a set $\{x_1,\hdots, x_n\} \subset \{-n, \hdots, n\}$ contains a pair $x_i, x_j$ such that $x_i - x_j = x_{i-j}$.
\end{itemize}

\subsection{Preliminaries and Definitions}

Let $S$ be a string of length $n$ over an alphabet $\Sigma = \{\sigma_1, \sigma_2, \hdots, \sigma_{|\Sigma|}\}$.  An integer $i$ is a {\em location} or a {\em position} in $S$ if $i \in \{1, \ldots, \abs{S}\}$. The substring
$S[i \ldotdot j]$ of $S$, for any two positions $i \leq j$, is
the substring of $S$ that begins at index $i$ and ends at index $j$. The string generated by a character $a$ repeated $r$ times is shorthanded with $a^r$.

The {\em Parikh vector} of a string $S$ is $\psi(S) = (c_1(S), c_2(S), \hdots, c_{|\Sigma|}(S))$, where $c_i(S)$ is the count of occurrences of the $i$-th character of $\Sigma$. Two strings (of equal length) $S$ and $S'$ are said to {\em jumble-match} if they have the same Parikh vector. For a text $T$ and pattern $P$ we say that $P$ {\em jumble-matches at location $i$} if the substring $T[i \ldotdot i+|P|-1]$ jumble-matches $P$. {\em Jumbled pattern matching} refers to the problem where one is given a pattern and text and seeks all locations where the pattern jumble-matches.
For a Parikh vector $\psi = (c_1,\hdots, c_{|\Sigma|})$, we denote its length with $|\psi|$ which is $\Sigma_{i=1}^{|\Sigma|} c_i$.

{\em Jumbled indexing} (JI, for short), also known as {\em histogram indexing}, {\em Parikh
indexing}, or  {\em permutation indexing}, is defined as follows.


\begin{itemize}
\item
{\bf Preprocess:} a text $S$ over alphabet $\Sigma$.
\item
{\bf Query:} Given a vector $\psi \in \mathbb{N}^{|\Sigma|}$, decide whether there is a substring $S'$ such that the Parikh vector $\psi(S')$ is equal to $\psi$.
\end{itemize}

\section{Hardness of Jumbled Indexing}

\subsection{Outline}


We will show that, under the 3SUM-hardness assumption, one cannot improve the running time over the naive methods mentioned in the introduction by any polynomial factors for alphabets of super-constant size;
and for alphabets of constant size there are polynomial time lower bounds,
dependent on the alphabet size.

To achieve these results we reduce from 3SUM to JI\@. Naturally, we use Convolution-3SUM, which is more appropriate for problems with structure. A very high-level description of our reduction is as follows. A Convolution-3SUM  input is transformed to JI by hashing the input values to much smaller sized values by using mod over a collection of primes. These are then novelly transformed to a string. The queries on the string simulate testing matchings mod primes in parallel. Using mod primes causes several problems, which lead to interesting ideas to overcome these obstacles.

\subsection{Setup}

Let $x_1,x_2,\hdots, x_n$ be the input of the Convolution-3SUM problem such that each $x_i \in \{-n^2,\hdots, n^2\}$. Under the strong 3SUM-hardness assumption, each $x_i \in \{-n,\hdots, n\}$.

We choose a collection of roughly equal-sized primes $p_1, \hdots, p_k$ (for some choice of $k$) with their product $p_1\cdots p_k > n^2$ (or, under the strong 3SUM-hardness assumption, with $p_1\cdots p_k > n$). It is possible to choose $p_1, \hdots, p_k \in \Theta(n^{2/k})$
 (or, for the strong assumption, $p_1, \hdots, p_k \in \Theta(n^{1/k})$)
to satisfy this requirement for any given $k \leq {\log n \over \log\log n}$, because of the density of the primes.

The alphabet of JI in the reduction will consist of a character for each prime we choose, plus two more special characters we introduce later. Therefore, the JI alphabet size will be $|\Sigma| = k+2$.

The lemma below follows directly from properties of mod and will be instrumental in obtaining our result.

\begin{lemma}~\label{primes}
Let $p_1, \hdots, p_k$ be a set of primes such that $p_1\cdots p_k > u$ (with $u=n^2$ or $u=n$ depending on the hardness assumption). Let $i>j$. Then $x_i - x_j = x_{i-j} \iff \forall r: (x_i - x_j) \ {\rm mod}\  p_r = x_{i-j} \ {\rm mod}\  p_r$

$\iff \forall r:  (x_i \ {\rm mod}\  p_r) - (x_j \ {\rm mod}\  p_r) \in \begin{cases} (x_{i-j} \ {\rm mod}\  p_r)\\(x_{i-j} \ {\rm mod}\  p_r) - p_r\end{cases}$
\end{lemma}

\subsection{Reduction}

In the reduction to the JI instance we will generate an input string $S$ to be preprocessed and a set of $n$ queries $Q_1, \hdots, Q_n$ which we now describe.

\subsubsection{JI Input String}

We generate an input string $S$ based on the Convolution-3SUM input $x_1, \hdots, x_n$. For every prime $p_j$ we create a character $a_j$ and for each $x_i$ we create a substring
$$S_i = a_1^{\textit{EXP}(i,1)}a_2^{\textit{EXP}(i,2)}\cdots a_k^{\textit{EXP}(i,k)},$$
where
$$\textit{EXP}(i,j) = (x_{i+1}\ {\rm mod}\ p_j)-(x_i\ {\rm mod}\ p_j).$$

We note that, for the sake of simplicity, we are cheating since the exponent of a character in a string cannot be negative. We will shortly explain how to fix this.

Finally, we define
$$S = \$\#\ S_1\ \#\$\#\ S_2\ \#\$\#\ \cdots\ \#\$\#\ S_{n-1}\ \#\$,$$
where $\#$ and $\$$ are separator characters.

\medskip

The structure of $S$ is such that substrings beginning and ending within separators $\#\$\#$ have the property that the number of occurrences of each character is reminiscent of the requirements of Lemma~\ref{primes}.

\begin{lemma}\label{telescope}
Consider the substring of $S$, $R_{(j,i)} = \$\#\ S_j\ \#\$\#\ \hdots\ \#\$\#\ S_{i-1}\ \#\$ $. Each character $a_{\ell}$ has exactly $(x_i\ {\rm mod}\ p_{\ell})\ -\ (x_j\ {\rm mod}\ p_{\ell})$ occurrences in $R_{(j,i)}$.
\end{lemma}
\begin{proof}
The character $a_{\ell}$ has $\textit{EXP}(j,\ell)$ occurrences in $S_j$, $\textit{EXP}(j+1,\ell)$ occurrences in $S_{j+1}, \hdots$,  $\textit{EXP}(i-1,\ell)$ occurrences in $S_{i-1}$. Hence, we have $\Sigma_{d=j}^{i-1}\textit{EXP}(d,\ell)$ occurrences of $a_{\ell}$ in $R_{(j,i)}$.  Then $\Sigma_{d=j}^{i-1}\textit{EXP}(d,\ell) =$ $\Sigma_{d=j}^{i-1}(x_{d+1}\ {\rm mod}\  p_{\ell})-(x_d \ {\rm mod}\  p_{\ell})$, which telescopes to $(x_i\ {\rm mod}\ p_{\ell})\ -\ (x_j\ {\rm mod}\ p_{\ell})$.
\qed
\end{proof}


By combining Lemmas~\ref{primes} and~\ref{telescope} we can deduce the following.

\begin{corollary}~\label{conclusion}
There is a solution $x_i - x_j = x_{i-j}$ to the Convolution-3SUM  iff
the number of occurrences of each character $a_{\ell}$
in $R_{(j,i)}$ is in
$$\{(x_{i-j}\ {\rm mod}\  p_{\ell}), (x_{i-j}\ {\rm mod}\  p_{\ell})-p_{\ell}\}.$$
\end{corollary}

To fix the problem of the negative exponent we set $D = \max_{i=1}^k p_i$ and change $\textit{EXP}(i,j) = (x_{i+1}\ {\rm mod}\ p_j)-(x_i\ {\rm mod}\ p_j) + D$. Now, the exponent is not negative, but is still of order $\Theta(n^{2/k})$ (or $\Theta(n^{1/k})$  under the strong 3SUM-hardness assumption), which is the size of each prime. We leave it as an easy exercise to verify that Lemma~\ref{telescope} can be modified so that each character $a_{\ell}$ has exactly $(x_i\ {\rm mod}\ p_{\ell})\ -\ (x_j\ {\rm mod}\ p_{\ell}) + D(i-j)$ occurrences in $R_{(j,i)}$ and, in turn, that Corollary~\ref{conclusion} can be modified so that  $a_{\ell} \in \{(x_{i-j}\ {\rm mod}\  p_{\ell}) + D(i-j), (x_{i-j}\ {\rm mod}\  p_{\ell})-p_{\ell} + D(i-j)\}$.

\subsubsection{JI Queries}

We generate $n$ queries $\psi_1, \hdots, \psi_n$ for the jumbled indexing instance such that each $\psi_L$ represents $x_L$, an element of the Convolution-3SUM  input. The query $\psi_L$ will imitate a query on the Convolution-3SUM  data asking whether there exist $i$ and $j$ such that
\begin{enumerate}
\item[(a)] $x_L = x_i - x_j$ and
\item[(b)] $L=i-j$.
\end{enumerate}

We will also embed the query with data requiring that
\begin{enumerate}
\item[(c)] any substring that jumble-matches the query $\psi$ must be of the form $R_{(j,i)}$ from Lemma~\ref{telescope}.
\end{enumerate}

Obviously, answers to all queries $\psi_L$ will be sufficient to derive a solution to Convolution-3SUM.

To enforce (c) and (b) we use the separators of $S$, $\#$, and $\$$. For (c) we require the form of $R_{(j,i)}$  and for (b) we require $L=i-j$, which means that each potential substring $R_{(j,i)}$ should contain exactly $L$ parts $S_h$.

\begin{observation}~\label{obs:query}
Any substring of $S$ that jumble-matches the query $\psi$, where $\psi$ has $L+1$ for $\$$ and $2L$ for $\#$, must be of the form $R_{(j,i)}$ (of Lemma~\ref{telescope}) and must satisfy $L=i-j$.
\end{observation}
\begin{proof}
Let $R$ be a substring of $S$ such that $R$ jumble-matches $\psi$. Then each set of separators $\#\$\#$ fully contained in $R$ contributes twice as many $\#$'s than $\$$'s. Since our query asks for $L+1$ $\$$'s but only $2L$ $\#$'s, it must be that $R$ begins and ends with a $\$$ and hence is of the form $R_{(j,i)}$. Moreover, since there are $L+1$ $\$$'s and $R$ begins and ends with a $\$$, there must be exactly $L$ parts $S_h$ in $R$, implying that $L = i-j$.
\qed
\end{proof}


It remains to show how to adapt the query in order to enforce (a) $x_L = x_i - x_j$. Here we will use Corollary~\ref{conclusion}. It is sufficient to find the substrings $R_{(j,i)}$ such the number of occurrences of each character $a_{\ell}$ is either $((x_i - x_j)\ {\rm mod}\  p_{\ell}) + DL$  or  $((x_i - x_j)\ {\rm mod}\  p_{\ell})-p_{\ell} +DL$. However, checking two options (for each $a_{\ell}$) cannot be done with one JI query. So, we split the query $\psi_L$ into $2^{k}$ queries $\psi_{L}^{(1)}, \hdots, \psi_{L}^{(2^k)}$ for the $2^k$ different equalities that satisfy Corollary~\ref{conclusion}.

Hence, we have overall $2^kn$ JI queries. These queries provide a full answer to the Convolution-3SUM problem.


\begin{theorem}~\label{thm:JI-Hardness}
Consider the jumbled indexing problem with text size $s$
and alphabet size $r\ge 5$. Then under the 3SUM-hardness assumption, one of the following holds
for any fixed $\epsilon>0$:

\begin{enumerate}
\item
the preprocessing time is $\Omega(s^{2-{4 \over r}-\epsilon})$, or
\item
the query time is $\Omega(s^{1-{2 \over r}-\epsilon})$.
\end{enumerate}

\end{theorem}

\begin{proof}
Without loss of generality, assume that $r\le {\log s\over\log\log s}$ (otherwise, $s^{1\over r}=\tilde{\Theta}(1)$ and we may as well make $r$ equal to ${\log s\over\log\log s}$).

Let $x_1,\ldots,x_n \in \{-n^2,\hdots, n^2\}$ be the input of the Convolution-3SUM problem. We apply the above reduction and generate the string $S$ as described. Denote its length by $s$.  Recall that the alphabet size is $r=|\Sigma|=k+2$, where $k$ is the number of primes (the 2 is for the separators $\$$ and $\#$).
Since each prime $p_i \in\Theta(n^{2/k})$,  we have $s=O(kn^{2/k}n) = \tilde{O}(n^{r \over r-2})$ for $k+2=r\in o(\log n)$.  In other words, $n=\tilde{\Omega}(s^{1-{2\over r}})$.

Applying the preprocessing and subsequently answering all $2^kn$ defined queries yields a solution to the Convolution-3SUM  problem.
Letting $P(s)$ and $Q(s)$ be the preprocessing and query time, we then have
$P(s)+2^k n Q(s) \ge \Omega(n^{2-\epsilon})$.

For $k+2=r \in o(\log n)$ we note that $2^k \in o(n^{\epsilon})$.  We must thus have $P(s)\ge\Omega(n^{2-\epsilon})=\Omega(s^{2(1-{2\over r})-O(\epsilon)})$ or $Q(s)\ge\Omega(n^{1-O(\epsilon)})=\Omega(s^{1-{2\over r}-O(\epsilon)})$.
%
%
%
\qed
\end{proof}

Note that for $r\le 4$, the bound in the above theorem becomes vacuous.
We can get somewhat better bounds under the strong 3SUM-hardness assumption.

\begin{theorem}\label{thm3}
Consider the jumbled indexing problem with text size $s$
and alphabet size $r\ge 4$. Then under the strong 3SUM-hardness assumption, one of the following holds for any fixed $\epsilon>0$:

\begin{enumerate}
\item
the preprocessing time is $\Omega(s^{2-{2 \over r-1}-\epsilon})$, or
\item
the query time is $\Omega(s^{1-{1\over r-1}-\epsilon})$.
\end{enumerate}
\end{theorem}

The proof is the same as in the previous theorem, but with $p_i \in \Theta(n^{1/k})$.

By the same proof and further calculations,
we can also get a slightly strengthened lower bound of
$\Omega(s^2/2^{O(\sqrt{\log s})})$ preprocessing time or $\Omega(s/2^{O(\sqrt{\log s})})$
query time for alphabet size $r=\Theta(\sqrt{\log s})$, under the assumption that 3SUM has an $\Omega(n^2/2^{O(\sqrt{\log n})})$ lower bound.




%

\subsection{Hardness of JI with Alphabet Size 3}

The reduction we have presented contains two separators in the string $S$. Recall that for every prime we also construct a character. We require that the multiplication of the primes be $>n$ for strong 3SUM-hardness and $>n^2$ for 3SUM-hardness. However, if a prime is of order $\Omega(n)$ then the size of the string $S$ would be $\Omega(n^2)$, too large to gain anything from the reduction. Hence, we need at least two primes for strong 3SUM-Hardness and three primes for 3SUM-hardness. In this section we generate a string which requires only one separator, and for 2 primes $p$ and $q$ of size $\Theta(\sqrt{n})$ this yields a nontrivial result
 under the strong 3SUM-hardness assumption.

While we construct a different string for JI and need to argue a claim similar to Lemma~\ref{telescope} and Observation~\ref{obs:query}, the structure of the proof remains the same.

Let $a$ be a character representing prime $p$, and $b$ be a character that represents prime $q$, and $\#$ be a separator character.  Let $D=\max\{p,q\}$.
Define

\medskip

$S_i = (a\#)^{(x_{i+1}\ {\rm mod}\ p)-(x_i\ {\rm mod}\ p) +D}(b\#)^{(x_{i+1}\ {\rm mod}\ q)-(x_i\ {\rm mod}\ q) +D}$ and

$S = \#^Da^{2D}\#^D\ S_1\ \#^Da^{2D}\#^D\ S_2\ \#^Da^{2D}\#^D\ \cdots\ \#^Da^{2D}\#^D\ S_{n-1}\ \#^Da^{2D}\#^D,$
where $\#^Da^{2D}\#^D$ is the separator ($a$ has a double role).

\medskip

Define $R_{(j,i)} = \#^D\ S_j\ \#^Da^{2D}\#^D\ \cdots\ \#^Da^{2D}\#^D\ S_{i-1}\ \#^D$.

\medskip

It is easy to verify, similar to Lemma~\ref{telescope}, that $a$ has exactly $(x_i\ {\rm mod}\ p)\ -\ (x_j\ {\rm mod}\ p) + D(i-j) + 2D(i-j-1)$ occurrences in $R_{(j,i)}$ and $b$ has exactly  $(x_i\ {\rm mod}\ q)\ -\ (x_j\ {\rm mod}\ q) + D(i-j)$ occurrences in $R_{(j,i)}$. Hence, as in Corollary~\ref{conclusion}, there is a solution $x_i - x_j = x_{i-j}$ to  Convolution-3SUM  iff the number of occurrences of $a$ in $R_{(j,i)}$  is in $\{((x_i - x_j)\ {\rm mod}\  p) + D(3i-3j-2), ((x_i - x_j)\ {\rm mod}\  p)-p + D(3i-3j-2)\}$ and the number of occurrences of $b$ in $R_{(j,i)}$ is in $\{((x_i - x_j)\ {\rm mod}\  p) + D(i-j), ((x_i - x_j)\ {\rm mod}\  p)-p + D(i-j)\}$.

\medskip

The tricky part is to obtain an alternative to Observation~\ref{obs:query}. The difficulty stems from the fact that $a$ and $\#$ appear both in the separator part and in the~$S_i$'s.

\begin{observation}~\label{obs:three}
Say we have a Parikh vector $\psi = (n_1, n_2, n_1+n_2+4D)$  for $(a,b,\#)$, with $L=(n_1\ \mbox{\em div}\ 3D) + 1$. Then any substring of $S$ which jumble-matches $\psi$ is of the form $R_{(j,i)}$. Moreover $i-j = L$.
\end{observation}
\begin{proof}
Define $\Delta = 4D$ the difference between the number of $\#$'s and the number of $a$'s and $b$'s (put together). Let $x$ be a substring of $S$ which jumble-matches~$\psi$. Any $S_l$ or separator $\#^Da^{2D}\#^D$ fully contained in $x$ has a balanced number of $\#$'s and non-$\#$'s and, hence, does not affect $\Delta$. So, at either end there is a part of a separator or an $S_l$ which both together  contributes to $\Delta$. It is straightforward to confirm that the only way this is possible is having $x$ begin with $\#^D$, the prefix of a separator, and end with $\#^D$, the suffix of a separator. Hence, $x$ has the required form $R_{(j,i)}$.

We show that  $i-j = L = (n_1$ div $3D)+1$. We have claimed that the number of $a$'s in $R_{(j,i)}$ is $(x_i\ {\rm mod}\ p)\ -\ (x_j\ {\rm mod}\ p) + D(3i-3j-2)$. Hence, since $R_{(j,i)}$ jumble-matches $\psi$, we have $n_1 = (x_i\ {\rm mod}\ p)\ -\ (x_j\ {\rm mod}\ p) + D+3D(i-j-1)$. Since each $(x_i\ {\rm mod}\ p)\ -\ (x_j\ {\rm mod}\ p) +D\in [0, 2D)$ it follows that $(n_1$ div $3D) = i-j-1$.
\qed
\end{proof}


Finally, for a given $L$ all of our queries have $n_1 = (x_L\ {\rm mod}\  p) + D(3L-2)$ or $n_1 = (x_L\ {\rm mod}\  p) - p + D(3L-2)$ in location $a$ of the Parikh vector. In both cases, $L=(n_1$  div $3D)+1$, satisfying the requirement of Observation~\ref{obs:three}. Hence,

\begin{theorem}
Consider the jumbled indexing problem with text size $s$ and alphabet size $3$. Under the strong 3SUM-hardness assumption, one of the following holds for any fixed $\epsilon>0$:

\begin{enumerate}
\item
the preprocessing time is $\Omega(s^{{4 \over 3}-\epsilon})$, or
\item
the query time is $\Omega(s^{{2 \over 3}-\epsilon})$.
\end{enumerate}
\end{theorem}

\begin{proof}
Note that $p,q \in \Theta(\sqrt{n})$. Hence, $S$ is of length $s=O(n^{3\over 2})$. Following the same arguments as in Theorem~\ref{thm:JI-Hardness} yields the result. \qed
\end{proof}

\paragraph{Epilogue.}
In a forthcoming work, the second and third author will present new improved algorithms for
the jumbled indexing problem for any constant alphabet size $r\ge 2$ that achieves truly sublinear query time and
$O(n^{2-{2\over r+O(1)}})$ preprocessing time, thus nearly matching the lower bound in
Theorem~\ref{thm3}.

\bibliographystyle{plain}
\bibliography{Reduction3SUM}




\end{document}